\documentclass{article}
\usepackage[T1]{fontenc}
\usepackage{xspace}
\usepackage[colorlinks=true,linkcolor=blue,citecolor=blue]{hyperref}
\usepackage{amssymb,amsmath,amsthm,graphicx,color}
\usepackage{physics, authblk}
\usepackage{mathtools}
\usepackage{enumitem}
\usepackage[margin=1.in]{geometry}
\usepackage{multirow,array}
\allowdisplaybreaks

\usepackage[capitalise,nameinlink]{cleveref}

\usepackage{algorithm}
\usepackage[noend]{algpseudocode}

\theoremstyle{plain}
\newtheorem{theorem}{Theorem}
\newtheorem{corollary}[theorem]{Corollary}
\newtheorem{proposition}[theorem]{Proposition}
\newtheorem{lemma}[theorem]{Lemma}
\newtheorem{claim}[theorem]{Claim}
\newtheorem{fact}[theorem]{Fact}

\newtheorem{definition}[theorem]{Definition}

\theoremstyle{plain}

\newcommand{\complexityclass}[1]{{\mathsf{#1}}\xspace}

\newcommand{\NP}{\complexityclass{NP}}

\newcommand{\BQP}{\complexityclass{BQP}}

\newcommand{\MIP}{\complexityclass{MIP}}
\newcommand{\QMIP}{\complexityclass{QMIP}}
\newcommand{\RE}{\complexityclass{RE}}
\newcommand{\SBQP}{\complexityclass{SBQP}}

\newcommand{\QMA}{\complexityclass{QMA}}

\newcommand{\EXP}{\complexityclass{EXP}}
\newcommand{\NEXP}{\complexityclass{NEXP}}
\newcommand{\SQMA}{\complexityclass{SQMA}}
\newcommand{\PP}{\complexityclass{PP}}
\newcommand{\PSPACE}{\complexityclass{PSPACE}}

\newcommand{\stateQMA}{\complexityclass{stateQMA}}

\renewcommand{\Pr}{\mathop{\bf Pr\/}}
\newcommand{\val}{\mathop{\bf val\/}}

\newcommand{\Adj}{\mathop{Adj\/}}

\newcommand{\poly}{\mathrm{poly}}

\newcommand{\polylog}{\mathrm{polylog}}

\newcommand{\eps}{\varepsilon}

\renewcommand{\tilde}{\widetilde}

\newcommand{\calA}{\mathcal{A}}
\newcommand{\calB}{\mathcal{B}}
\newcommand{\calC}{\mathcal{C}}

\newcommand{\calM}{\mathcal{M}}

\newcommand{\GapCSP}{\textsc{GapCSP}}

\begin{document}
\title{Quantum Merlin-Arthur and proofs without relative phase}
 \author[]{Roozbeh Bassirian\thanks{\href{mailto:roozbeh@uchicago.edu}{roozbeh@uchicago.edu}}}
 \author[]{Bill Fefferman\thanks{\href{mailto:wjf@uchicago.edu}{wjf@uchicago.edu}}}
 \author[]{Kunal Marwaha\thanks{\href{mailto:kmarw@uchicago.edu}{kmarw@uchicago.edu}}}
 \affil{University of Chicago}
\date{}
\maketitle

\begin{abstract}
We study a variant of $\QMA$ where quantum proofs have no relative phase (i.e. non-negative amplitudes, up to a global phase).
If only completeness is modified, this class is equal to $\QMA$~\cite{sqma}; but if both completeness and soundness are modified, the class (named $\QMA^+$ by Jeronimo and Wu~\cite{qma2plus}) can be much more powerful.
We show that $\QMA^+$ with some constant gap is equal to $\NEXP$, yet $\QMA^+$ with some \emph{other} constant gap is equal to $\QMA$. 
One interpretation is that Merlin's ability to ``deceive'' originates from \emph{relative phase} at least as much as from \emph{entanglement}, since $\QMA(2) \subseteq \NEXP$.
\end{abstract}

\section{Introduction}
The strangeness of quantum states has at least two fundamental sources: \emph{entanglement}, the source of ``spooky action at a distance''; and \emph{relative phase}, which allows for destructive interference. We use complexity theory to probe these sources of strangeness. Extending the main result of \cite{qma2plus}, we find that $\QMA^+$ ($\QMA$ where quantum proofs have no relative phase) is as powerful as $\NEXP$.

A $\QMA_{c,s}$ protocol is a verification task for a quantum computer (termed ``Arthur'') when interacting with a dishonest but all-powerful machine (termed ``Merlin'').
If the statement is true (``completeness''), Merlin sends a quantum state (``proof'') that truthfully convinces Arthur. If the statement is false (``soundness''), Merlin will send any quantum state possible to deceive Arthur. A valid protocol distinguishes these cases, succeeding with probability at least $c$ in completeness and at most $s < c$ in soundness. Canonically, $\QMA$ is the class of all valid $\QMA_{2/3, 1/3}$ protocols.

One could potentially reduce the power of $\QMA$ by restricting Merlin's proof in completeness. Surprisingly, many restrictions of this type \emph{do not reduce the power of the class}. For example, this is true even if the quantum state is a \emph{subset state} (with no relative phase nor relative non-zero amplitude) \cite{sqma}. The reason behind this is \emph{promise gap amplification}: there exist techniques to increase the gap $c-s$ to $1 - 2^{-p(n)}$ for any polynomial $p(n)$. As a result, a subset state with polynomially small overlap with the best completeness proof succeeds. This argument generalizes to any set of states that form an $\epsilon_n$-covering of all $n$-qubit quantum states, where $\epsilon_n$ is at least inverse polynomial in $n$.

By contrast, restricting Merlin's proof in soundness seems to increase the power of this complexity class, since this reduces Merlin's ability to ``deceive''. For example, if Merlin must send a quantum proof \emph{without relative phase}, Arthur can ask about its \emph{sparsity} ($\ell_1$ norm).
When a state has no relative phase, a low overlap with $\ket{+}^{\otimes n}$ actually implies it is sparse, as opposed to a state with destructively interfering phases (i.e. any other Hadamard basis vector).

One popular variant of $\QMA$ restricts Merlin's \emph{entanglement} over a fixed barrier; it is named $\QMA(2)$ (as if there are two unentangled Merlins each sending a quantum proof \cite{qma2_defn}). This complexity class may seem more powerful than $\QMA$, but despite much study~\cite{blier2010quantum,aaronson2008power,chen2010short, pereszlényi2012multiprover,harrow2013testing,qma2_yirka,she2022unitary}, little is known except the trivial bounds $\QMA \subseteq \QMA(2) \subseteq \NEXP$.

What happens if one restricts both \emph{entanglement} and \emph{relative phase}?
\cite{qma2plus} define $\QMA^+_{c,s}$ and $\QMA^+(2)_{c,s}$, where quantum proofs are required to have no relative phase (non-negative amplitudes, up to a global phase) \emph{in both cases}.\footnote{As noted before, restricting the state in completeness may not change the complexity class, but restricting the state in soundness can make the class more powerful, since the latter limits Merlin's adversarial strategies.} Surprisingly, \cite{qma2plus} show the existence of constants $1 > c > s > 0$ such that $\QMA^+(2)_{c,s} = \NEXP$, crucially including a protocol to estimate the \emph{sparsity} of a quantum proof.
This hints perhaps at a route to prove $\QMA(2) = \NEXP$, since there are other constants $1 > c' > s' > 0$ where $\QMA^+(2)_{c',s'} = \QMA(2)$.\footnote{This is because every state has constant overlap with some state without relative phase. See also~\Cref{prop:qmaplus_nearly_qma}.}

In this work, we show that restricting \emph{relative phase} alone gives the power of $\NEXP$; i.e., there exist constants $1 > c > s > 0$ where $\QMA^+_{c,s} = \NEXP$. Note that assuming $\EXP \ne \NEXP$, this implies $\QMA^+$ cannot be amplified, since as before, there are other constants  $1 > c' > s' > 0$ where $\QMA^+_{c',s'} = \QMA \subseteq \EXP$.
As a result, techniques to prove $\QMA(2) = \QMA^+(2) = \NEXP$ must crucially use the unentanglement promise inherent in $\QMA(2)$. See \Cref{fig:qma2_vs_qmaplus} and \Cref{fig:qma_vs_nexp} for a pictorial description.

\begin{figure}[t]
\centering
\includegraphics[width=0.5\textwidth]{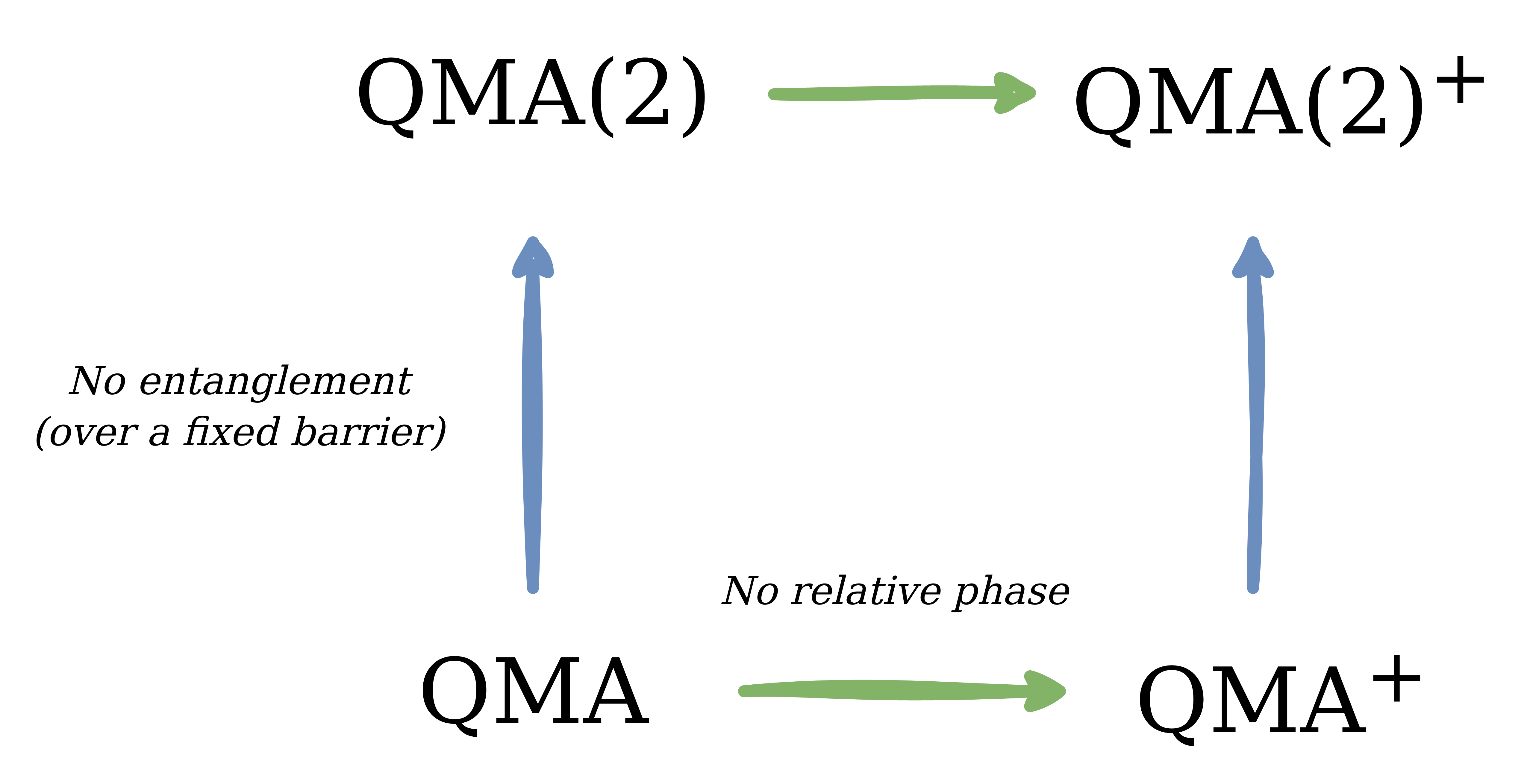}
\caption{\footnotesize Relationship between $\QMA^+$ and $\QMA(2)$. \cite{qma2plus} show that for some constants $1 > c > s > 0$, $\QMA^+(2)_{c,s} = \NEXP$. We show that the same is true for $\QMA^+$. 
Restricting relative phase does not restrict entanglement across a fixed barrier: for example, consider the GHZ state $\ket{0}^{\otimes n} + \ket{1}^{\otimes n}$, or more generally states where the Schmidt vectors have no relative phase.}
\label{fig:qma2_vs_qmaplus}
\end{figure}

\subsection{Techniques}
Our primary technical contribution is to show a $\QMA^+$ protocol for a $\NEXP$-complete problem. This directly extends the work of Jeronimo and Wu~\cite{qma2plus}, who show a $\QMA^+(2)$ protocol for a $\NEXP$-complete problem. As in~\cite{qma2plus}, we study constraint satisfaction problems (CSPs) with \emph{constant} gap. In $(1, \delta)$-$\GapCSP$, either \emph{all} constraints can be satisfied, or at most a $\delta$ fraction of constraints can be satisfied.
These problems are known to be $\NP$-hard or $\NEXP$-hard (depending on the problem size) using the PCP theorem~\cite{AS92, ALMSS98, harsha2004robust}.

Before proving $\QMA^+$ with some constant gap equals $\NEXP$, we prove $\QMA^+_{\log}$ (with some other constant gap) equals $\NP$. This choice (also taken by \cite{qma2plus}) is pedagogical: it allows us to explain the protocol without worrying about input encoding size, since $\QMA^+$ has a polynomial amount of space and verifier runtime. Here, we consider $(1, \delta)$-$\GapCSP$ with polynomially many variables and clauses; the quantum proof must certify that there is a satisfying assignment to all clauses.

The $\QMA^+_{\log}(2)$ protocol of \cite{qma2plus} crucially relies on an estimate of \emph{sparsity} ($\ell_1$ norm) of a quantum state without relative phase.
The overlap of a $m$-qubit quantum state without relative phase $\ket{\psi}$ with $\ket{+}^{\otimes m}$ is exactly the value $2^{-m/2} \cdot \|\ket{\psi}\|_1$. With multiple quantum proofs $\ket{\psi_1} \otimes \dots \ket{\psi_k}$, one can \emph{estimate} the sparsity by repeating this ``sparsity test'' on each $\ket{\psi_i}$, and using a swap test to ensure that all $\ket{\psi_i}$ are approximately equal. Interestingly, no other part of their protocol requires the \emph{no relative phase} assumption.\footnote{Formally, \cite{qma2plus} studies states with non-negative amplitudes. Recall that the set of these states, up to global phase, are equivalent to states with \emph{no relative phase}.} 

In $\QMA^+_{\log}$, we have a single quantum proof, so we cannot use this test to estimate sparsity. Instead, we design a similar test that directly enforces a \emph{rigidity} property of the proof.\footnote{Note that we use the intuition of \emph{rigidity} in a more general context, where Arthur's tests, not a non-local game, enforce states of a certain form.} The required form is $\frac{1}{\sqrt{R}} \sum_{j \in [R]} \ket{j}\ket{\vec{v}_j}$, where the second register is constant-sized. Using a ``sparsity test'' over the second register, Arthur ensures that the second register has one $\vec{v}_j$ per $j$; but using the \emph{complement} of a ``sparsity test'' over the whole proof, Arthur ensures the overall state maximizes $\ell_1$ norm. States of the required form are optimal for this combination of tests. We make use of the \emph{no relative phase} property in \Cref{lemma:validity_upper_bound,lemma:non-negative_valid}.

Now we can describe our protocol. For each constraint $j \in R$, Arthur asks for the values $\vec{v}_j$ associated with the variables involved in constraint $j$. The protocol either enforces \emph{rigidity} of the quantum proof, or verifies the \emph{constraints} of the CSP. 
Note that we need two kinds of constraint checks: the values $\vec{v}_j$ must \emph{satisfy} constraint $j$, and the value of a variable must be \emph{consistent} across the constraints it participates in. For states with the rigidity property, checking \emph{satisfiability} is simple: measure in the computational basis and verify the measured constraint $j,\vec{v}_j$. States with the rigidity property will succeed with probability equal to the satisfying fraction of the CSP assignment.

Checking \emph{consistency} is done using a technique called ``regularization'' from the PCP literature~\cite{dinur2007pcp}; for each constraint $j$, we verify that each variable participating in $j$ has the same value in exactly $d$ other constraints for some constant $d$, in a way that the edges form an \emph{expander graph}. The expansion property guarantees that cheating on this test is as damaging as cheating on the satisfiability test. Jeronimo and Wu~\cite{qma2plus} use a swap test to implement these new checks, but this requires multiple quantum proofs. We show how to use a Hadamard test (which requires only one quantum proof) to achieve the same result, building on ideas from previous work~\cite{bassirian2022power}.
Since there exists a $\delta$ such that $(1, \delta)$-$\GapCSP$ is $\NP$-hard, this completes the proof of $\NP \subseteq \QMA^+_{\log}$ with some constant gap.

When scaling up to $\QMA^+$, one must be careful of how to succinctly encode the input of a $\NEXP$-complete problem. The PCP theorem allows us to choose $(1, \delta)$-$\GapCSP$ that is \emph{succinct}, but we need stronger properties.
Following the adjustments taken in \cite{qma2plus}, we choose a PCP system for $\NEXP$ that is both \emph{doubly explicit} and \emph{strongly uniform}. 
\emph{Doubly explicit} means that one can efficiently compute the variables participating in a given constraint \emph{and} the constraints a given variable participates in; using this, we can implement the consistency checks in polynomial time. \emph{Strongly uniform} means that the number of constraints a variable participates in is efficiently computable, and one of a fixed number of possibilities; using this, we only need to build a fixed number of expander graphs during regularization. Recent work also shows how to construct exponentially-sized expander graphs in polynomial time~\cite{lubotzky2009finite,alon2021explicit}. Once we are through these input encoding difficulties, our protocol is identical to that for $\NP$.

Fundamentally, the \emph{no relative phase} property allows Arthur to verify a number of constraints exponential in the number of qubits.
Attempts to do this for $\QMA(2)$ gave too small a promise gap~\cite{blier2010quantum, pereszlényi2012multiprover, GallNN12}, too many provers~\cite{aaronson2008power, chen2010short, chiesa2011improved}, or too much space or time~\cite{harrow2013testing, natarajan2023quantum}.
Jeronimo and Wu~\cite{qma2plus} show that $\QMA^+(2)$ circumvents this difficulty: using \emph{no relative phase} and \emph{unentanglement},  Arthur enforces the \emph{sparsity} of a quantum proof to solve a $\NEXP$-complete problem. At the center of our work is the insight that \emph{no relative phase} is enough for Arthur to require constant-sized answers to exponentially many questions, solving a $\NEXP$-complete problem with a single polynomial-size quantum proof.

\subsection{Related work}
\paragraph{The complexity class $\QMA(2)$} 
The complexity class $\QMA(2)$ is known to have promise gap amplification, and to be equal to $\QMA(k)$ for any $k$ at most polynomial in $n$~\cite{harrow2013testing}.
It is not obvious how to test for entanglement; even determining whether a polynomially-sized vector is entangled is $\NP$-hard \cite{gharibian2009strong}.
If there exist efficient approximate ``disentanglers'' that can create any separable state, then $\QMA = \QMA(2)$; see~\cite{aaronson2008power} for some progress. \cite{qma2_yirka} describe quantum variants of the polynomial hierarchy and connect their properties to bounds on $\QMA(2)$. 
It is not even known whether there is a \emph{quantum} oracle separating $\QMA$ and $\QMA(2)$~\cite{aaronson2021open}.

\begin{figure}[t]
\centering
\includegraphics[width=\textwidth]{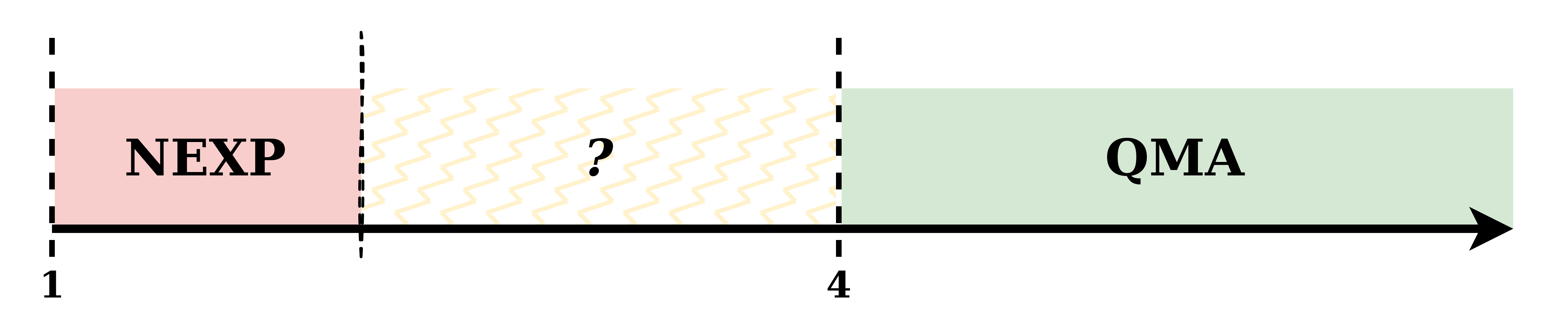}
\caption{\footnotesize Plot of $\QMA^+_{c,s}$ for increasing $\frac{c}{s}$. By our main result, there exist constants $c,s$ where $1 < \frac{c}{s} < 4$ and $\QMA^+_{c,s} = \NEXP$, but by \Cref{cor:qma_plus_equals_qma}, $\QMA^+_{c,s} = \QMA$ when $\frac{c}{s} > 4$. Gap amplification of $\QMA^+$ would imply $\QMA = \NEXP$.}
\label{fig:qma_vs_nexp}
\end{figure}

\paragraph{PCPs and expander graphs} Probabilistically checkable proofs (PCPs) show hardness for CSPs with a constant gap~\cite{AS92, ALMSS98, harsha2004robust}. Dinur~\cite{dinur2007pcp} proves the PCP theorem using a \emph{regularization} step, which adds new constraints associated with the edges of a \emph{regular} expander graph. \emph{Polynomial-time} regularization for $\NEXP$ requires an efficient description of exponentially-sized expander graphs. Recent advances in expander graph constructions~\cite{lubotzky2009finite,alon2021explicit} allow for this type of regularization, first used in \cite{qma2plus}.

\paragraph{Quantum states and relative phase} Up to a global phase, states with non-negative amplitudes are equivalent to states with no relative phase.
\cite{qma2plus} propose the class $\QMA^+$ and $\QMA^+(2)$, and show $\QMA^+(2) = \NEXP$. Note that by contrast, $\QMA$ restricted to states with \emph{real} amplitudes is equal to $\QMA$~\cite{McKAGUE_2013}.  Relative phase was recently proposed as a quantum resource~\cite{xu2023quantifying}. 
For both $\QMA$ and $\QMA(2)$, restricting Merlin in completeness to send a \emph{subset state} does not change the power of the complexity class (i.e., $\QMA = \SQMA$ and $\QMA(2) = \SQMA(2)$). \cite{sqma} also shows why their proof strategy fails if Merlin is restricted in both completeness and soundness.

\paragraph{Rigidity and games} Rigidity was first formally introduced in the context of non-local games \cite{mayers2004self}, and have been used to prove several complexity class equalities. 
For example, the CHSH game~\cite{chsh} tests for a maximally entangled state on two qubits~\cite{tsirelson1993some}, and was used to prove $\QMIP = \MIP^*$ \cite{qmip}. The Mermin-Peres magic square game tests for two copies of a maximally entangled quantum state, and was used to prove $\MIP^* = \RE$~\cite{mipre}. Rigidity is known to exist in broader contexts, including some (but not all) linear constraint games~\cite{Cui_2020} and monogamy-of-entanglement games~\cite{culf21}.


\section{Our setup}
We restate the definition of $\QMA^+(k)$ from \cite{qma2plus}. When the proof length is not specified, it is allowed to be at most any polynomial in input size. We follow the conventions $\QMA^+ := \QMA^+(1)$, $\QMA^+(k) := \bigcup_{c - s = \Omega(1)} \QMA^+(k)_{c,s}$, and $\QMA^+_{\log} := \QMA^+$ with proof length at most $O(\log n)$.
\begin{definition} [$\QMA^+_{\ell}(k)_{c,s}$]
Let $k: \mathbb{N} \rightarrow \mathbb{N}$ and $s, c, \ell: \mathbb{N} \rightarrow \mathbb{R}^+$ be polynomial time computable functions. A promise problem $\mathcal{L}_{\text{yes}}, \mathcal{L}_{\text{no}} \subseteq \{0, 1\}^*$ is in  $\QMA^+_{\ell}(k)_{c,s}$ if there exists a $\BQP$ verifier $V$ such that for every $n \in \mathbb{N}$ and every $x \in \{0, 1\}^n$:
\begin{itemize}
\item \textbf{Completeness:} if $x \in \mathcal{L}_{\text{yes}}$, then there exist unentangled states $\ket{\psi_1}, \ldots, \ket{\psi_{k(n)}}$, each on at most $\ell(n)$ qubits and with real non-negative amplitudes, s.t.\ 
\begin{align*}\Pr[V(x, \ket{\psi_1}\otimes \ldots \otimes \ket{\psi_{k(n)}}) \text{ accepts}] \geq c(n)\,.\end{align*}
\item \textbf{Soundness:} If $x \in \mathcal{L}_{\text{no}}$, then for every set of unentangled states $\ket{\psi_1}, \ldots, \ket{\psi_{k(n)}}$, each on at most $\ell(n)$ qubits and with real non-negative amplitudes, we have
\begin{align*}\Pr[V(x, \ket{\psi_1}\otimes \ldots \otimes \ket{\psi_{k(n)}}) \text{ accepts}] \leq s(n)\,.\end{align*}
\end{itemize}
\end{definition}

We make a few remarks on this complexity class, with extended discussion in \Cref{sec:qmaplus_subtleties}. First, we stress that the restriction to quantum proofs with non-negative amplitudes is \emph{promise-symmetric}, i.e. both in completeness and in soundness. This is unlike, for example, the class $\SQMA$~\cite{sqma}. Although the restriction to \emph{subset states} is stronger than \emph{non-negative amplitudes},\footnote{A \emph{subset state} is a uniform superposition over a subset of all computational basis states. States with non-negative amplitudes are conical combinations of subset states.} its use \emph{only} in completeness allows for $\SQMA = \QMA$. In fact, our work implies that $\QMA$ with a promise-symmetric \emph{subset state} restriction also interpolates from $\QMA$ to $\NEXP$, depending on the size of the promise gap.

We also explain why the promise gap of $\QMA^+$ cannot obviously be amplified. The first strategy one might try is \emph{parallel repetition}: an honest Merlin sends multiple copies of the original proof and Arthur verifies each copy of the original proof. For $\QMA$, entangling the copies in soundness does not help Merlin, since Arthur's protocol is sound for all quantum states. But perhaps unintuitively, it \emph{can} help for $\QMA^+$. (See \Cref{fact:max_of_reals_is_not_norm} for a simple example.) This is because partial measurement can \emph{destroy} the restriction on the quantum proof! For example, Arthur's first measurement may introduce relative phase in the rest of the proof.
This fact also obstructs more clever amplification strategies for $\QMA$ such as the proof-length preserving variant~\cite{marriott2005quantum}.

Furthermore, it is not clear how to upper-bound $\QMA^+$ beyond the trivial $\NEXP$.\footnote{$\QMA^+ \subseteq \NEXP$ by directly simulating the quantum proof and verifier.}
One technique to upper-bound $\QMA$ is to find the optimal proof using a semidefinite program (or a general convex program). This shows that $\QMA \subseteq \PSPACE$ (or $\QMA \subseteq \EXP$ with a convex program).
But these arguments do not immediately transfer to $\QMA^+$.
Convex optimization over states with non-negative amplitudes is equivalent to optimizing over the \emph{copositive cone}~\cite{burer2011copositive}.
Even the \emph{weak membership} problem over the copositive cone (deciding if the optimal vector is close to a non-negative vector) is $\NP$-hard in polynomially-sized vector spaces; recall that quantum states are in \emph{exponentially}-sized vector spaces.
These are the same reasons that prevent straightforward upper bounds for $\QMA(2)$~\cite{gharibian2009strong}.

\section{$\QMA^+_{\log}$ Protocol for $\NP$}

We first define the problem we consider:
\begin{definition}[CSP system]
    A \emph{$(N, R, q, \Sigma)$-CSP system} $\mathcal{C}$ on $N$ variables with values in $\Sigma$ consists of a set (possibly a multi-set) of $R$ constraints $\{\mathcal{C}_1, \dots, \mathcal{C}_R\}$ where the arity of each constraint is exactly $q$. 
\end{definition}
\begin{definition}[Value of CSP]
    The \emph{value} of a $(N, R, q, \Sigma)$-CSP system $\mathcal{C}$ is the maximum fraction of satisfiable constraints over all possible assignments $\sigma : [N] \to \Sigma$. The value of $\mathcal{C}$ is denoted $\val(\mathcal{C})$.
\end{definition}
\begin{definition}[$\GapCSP$]
    The \emph{$(1, \delta)$-$\GapCSP$ problem} inputs a CSP system $\mathcal{C}$. The task is to distinguish whether $\mathcal{C}$ is such that (in completeness) $\val(\mathcal{C}) = 1$ or (in soundness) $\val(\mathcal{C}) \le \delta$.
\end{definition}
Fix the input size $n$, and consider $(N,R,q,\Sigma)$-CSP systems where $N$ and $R$ are polynomials in $n$. Deciding whether or not these systems are satisfiable is $\NP$-hard. In fact, there exists $\delta < 1$ such that deciding $(1, \delta)$-$\GapCSP$ on these CSP systems is $\NP$-hard.
\begin{theorem}[\cite{dinur2007pcp}]
\label{thm:gapug_nphard}
There exist constants $q>1$ and $|\Sigma|$ such that $(1, 1/2)$-$\GapCSP$ is $\NP$-hard.
\end{theorem}
Our goal in this section is to construct a protocol for $(1, \delta)$-$\GapCSP$ given any $(N,R,q,\Sigma)$-CSP system $\mathcal{C}$ where $N,R = \poly(n)$ and $q, |\Sigma| = O(1)$.
Let $\kappa := |\Sigma|^q$. We first outline the protocol. Arthur asks for a quantum state from $\mathbb{C}^R \otimes \mathbb{C}^{\kappa}$; we call the first register 
the \emph{constraint register} and the second register the \emph{color register}. 
A quantum proof has the following form:
\begin{align*}\ket{\psi} := \sum_{j\in [R], x \in \Sigma^q} a_{j, x} \ket{j}\ket{x}\end{align*}
For completeness, consider the satisfying assignment of  variables (in $[N]$) to values (in $\Sigma$). Merlin sends the quantum proof $\frac{1}{\sqrt{R}}\sum_{j\in[R]}  \ket{j} \ket{\vec{v}_j}$, where each $\vec{v}_j$ is the (ordered) list of values associated with the variables participating in $\calC_j$.

Arthur then applies one of two kinds of tests:
\begin{enumerate}
    \item \emph{Rigidity tests}: These ensure that the quantum proof is of the form $\frac{1}{\sqrt{R}}\sum_{j\in[R]}  \ket{j} \ket{\vec{v}_j}$. 
    \item \emph{Constraint tests}: These verify that values in the quantum proof satisfy constraints of the CSP system.
\end{enumerate}
Below, we separately describe the rigidity tests and constraint tests. In each, we analyze the success probability in completeness and prove lemmas to study soundness. We then combine the technical statements to prove the result.

\subsection{Rigidity tests}

Arthur enforces \emph{rigidity} of a quantum proof using two tests.
The first test is the \emph{Density test}, which maximizes $\ell_1$ norm.
Here, we measure the state in the Hadamard basis and \emph{accept} if the outcome is $\ket{+}$.\footnote{For simplicity, we denote the uniform superposition over all standard basis states by $\ket{+}$. The dimension is clear from the context.} Given $\ket{\psi}$, the success probability of this test is
\begin{align*}
    D(\ket{\psi}) = \left|\bra{+}\ket{\psi}\right|^2 = \frac{1}{\kappa R}\left|\sum_{j \in [R], x \in \Sigma^q} a_{j,x}\right|^2 = \frac{1}{\kappa R} (\|\ket{\psi}\|_1)^2\,.
\end{align*}
Recall that if $\ket{\psi}$ is a subset state according to subset $S$, its sparsity $\|\ket{\psi}\|_1$ is exactly $\sqrt{|S|}$.
In completeness, the quantum proof is a subset state with $R$ elements, so this test passes with probability $\frac{1}{\kappa}$.

The second test is the \emph{Validity test}, which minimizes $\ell_1$ norm \emph{only} on the second register.
Here, we measure the color register in the Hadamard basis, and \emph{reject} if the outcome is $\ket{+}$. Given $\ket{\psi}$, the success probability of this test is 
    \begin{align*}
        V(\ket{\psi}) = 1 - \bra{+}\Tr_{R}(\ket{\psi}\bra{\psi})\ket{+} = 1 - \frac{1}{\kappa} \sum_{j \in [R]} \left|\sum_{x \in \Sigma^q} a_{j,x}\right|^2\,,
    \end{align*}
where $\Tr_{R}$ is partial trace over the constraint register.
In completeness, recall that the proof has the form $\frac{1}{\sqrt{R}}\sum_{j \in [R]} \ket{j} \ket{\vec{v}_j}$, so the success probability is
\begin{align*}
    1 - \bra{+}
     \left( \frac{1}{R} \sum_{j \in [R]} \ketbra{\vec{v}_j}{\vec{v}_j} 
     \right) 
     \ket{+} = 1 - \frac{1}{\kappa}\,.
\end{align*}
In fact, no quantum state without relative phase can pass the \emph{Validity test} with a higher probability:
\begin{lemma}
\label{lemma:validity_upper_bound}
Suppose $\ket{\psi}$ has no relative phase. Then $V(\ket{\psi}) \le 1 - \frac{1}{\kappa}$.
\end{lemma}
\begin{proof}
The success probability $V(\ket{\psi})$ is
\begin{align*}
         1 - \frac{1}{\kappa} \sum_{j \in [R]} (\sum_{x \in \Sigma^q} a_{j,x})^2
        &= 
        1 - \frac{1}{\kappa} 
        (\sum_{j \in [R]} \sum_{x,y \in \Sigma^q} a_{j,x} a_{j,y})
        &= 
        1 - \frac{1}{\kappa} 
        (1 + \sum_{j \in [R]} \sum_{x,y \in \Sigma^q; x \ne y} a_{j,x} a_{j,y})
         &\le 1 - \frac{1}{\kappa}\,,
\end{align*}
where the second equality follows from $\sum_{j,x} a_{j,x}^2 = 1$ and the inequality holds since $a_{j,x} \ge 0$.
\end{proof}
It turns out that is impossible to score high on both the \emph{Validity test} and the \emph{Density test}.
We use this to enforce the rigidity property of $\ket{\psi}$.
\begin{lemma} \label{lemma:density_validity}
    $D(\ket{\psi}) + V(\ket{\psi}) \le 1$.
\end{lemma}
\begin{proof}
By Cauchy-Schwarz,
\begin{align*}
    D(\ket{\psi}) 
    = \frac{1}{\kappa R} \left| \sum_{j \in [R]} \sum_{x \in \Sigma^q} a_{j,x} \right|^2 \le \frac{1}{\kappa} \sum_{j \in [R]} \left|\sum_{x \in \Sigma^q} a_{j,x}\right|^2 = 1 - V(\ket{\psi})\,.\tag*{\qedhere} 
    \end{align*}
\end{proof}
Why does this help with rigidity?
Suppose Arthur inputs a quantum proof (without relative phase) $\ket{\psi}$ and runs \emph{Density test} with probability $p_1$ and \emph{Validity test} with probability $p_2$. Suppose also that $p_2 > p_1$. Then the expected success probability is $p_1 D(\ket{\psi}) + p_2 V(\ket{\psi}) \le p_1 + (p_2 - p_1) V(\ket{\psi}) \le p_1 + (p_2 - p_1)(1 - \frac{1}{\kappa})$. Note that this upper bound is achieved in completeness, and for any state of the form $\frac{1}{\sqrt{R}}\sum_{j \in R}\ket{j}\ket{\vec{v}_j}$. We show that quantum proofs must have this form to reach the upper bound.

One requirement to get close to the upper bound is near-optimal success probability on \emph{Validity test}. We prove that any quantum proof that has this property must be close to a state that assigns one color to each constraint.
\begin{lemma} \label{lemma:non-negative_valid}
Given $\ket{\psi} = \sum_{j,x} a_{j,x} \ket{j}\ket{x}$ with no relative phase (i.e. $a_{j,x} \ge 0$), let 
\begin{align*}
\gamma := \max_{\nu: [R] \to \Sigma^q} \sum_{j \in [R]} a_{j, \nu(j)}^2
\end{align*}
be associated with maximizing function $\sigma$, and let  $\ket{\phi} := \frac{1}{\sqrt{\gamma}} \sum_{j} a_{j,\sigma(j)} \ket{j}\ket{\sigma(j)}$. Fix any $d \ge 0$. If $V(\ket{\psi}) =1 - \frac{1 + d}{\kappa}$, then $\left|\bra{\psi}\ket{\phi}\right|^2 \ge 1 - d$. 
\end{lemma}
\begin{proof}
   Note that for all $x \in \Sigma^q$, $a_{j,\sigma(j)} \ge a_{j,x}$; otherwise, $\sigma$ is not maximizing. Using the proof of \Cref{lemma:validity_upper_bound},
    \begin{align*}
      d = \sum_{j \in [R]} \sum_{x, y \in \Sigma^q; x \ne y} a_{j,x} a_{j,y} \ge \sum_{j \in [R]} \sum_{y \in \Sigma^q; y \ne \sigma(j)} a_{j,\sigma(j)} a_{j,y} \ge \sum_{j \in [R]} \sum_{y \in \Sigma^q; y \ne \sigma(j)} a_{j,y}^2\,.
    \end{align*}
    So then,
    \begin{align*}
        \gamma = \sum_{j \in [R]} a_{j,\sigma(j)}^2 = (\sum_{j \in [R]} \sum_{x \in \Sigma^q} a_{j,x}^2) -  (\sum_{j \in [R]} \sum_{x \in \Sigma^q; x \ne \sigma(j)} a_{j,x}^2) \ge 1 - d\,.
    \end{align*}
    So $\left|\bra{\psi}\ket{\phi}\right|^2 = (\frac{1}{\sqrt{\gamma}} \sum_{j} a_{j,\sigma(j)}^2)^2 = \gamma \ge 1 - d$.
\end{proof}

Another requirement to get close to the upper bound is near-optimal success probability on \emph{Density test}, up to \Cref{lemma:density_validity}. Consider any quantum proof that passes \emph{Validity test} with probability close to $1 - \frac{1}{\kappa}$ and \emph{Density test} with probability close to $\frac{1}{\kappa}$; we show it must be close to a state of the form $\frac{1}{\sqrt{R}}\sum_{j \in R}\ket{j}\ket{\vec{v}_j}$. Now we can prove the soundness of the rigidity test by relying on the following fact:
\begin{fact} \label{fact:preserv}
Let $0\leq \Pi \leq \mathbb{I}$ be a positive semi-definite matrix, and let $\ket{\psi_1}$ and $\ket{\psi_2}$ be quantum states such that $\left|\braket{\psi_1}{\psi_2}\right|^2 \geq 1-d$. Then $\left|\bra{\psi_1}\Pi\ket{\psi_1} - \bra{\psi_2}\Pi\ket{\psi_2} \right| \leq \sqrt{d}$.
\end{fact}
\begin{proof}
    The quantity $\left|\bra{\psi_1}\Pi\ket{\psi_1} - \bra{\psi_2} \Pi\ket{\psi_2} \right| 
    = \left|\Tr(\Pi \left( \ketbra{\psi_1} - \ketbra{\psi_2} \right))\right|$
    is upper-bounded by the trace distance of $\ketbra{\psi_1}$ and $\ketbra{\psi_2}$, which has value $\sqrt{1 - \left|\bra{\psi_1}\ket{\psi_2}\right|^2} \le \sqrt{d}$.
\end{proof}
\begin{lemma}[Rigidity lemma] \label{lemma:validform}
    Let $d_2 \geq d_1 \geq 0$ be small constants. Suppose $\ket{\psi}$, $\ket{\phi}$, and $\sigma$ are defined as in \Cref{lemma:non-negative_valid}, and $\ket{\chi}$ is defined as
    \begin{align*}\ket{\chi} := \frac{1}{\sqrt{R}} \sum_{j \in [R]} \ket{j}\ket{\sigma(j)}\,.
    \end{align*}
    If $D(\ket{\psi}) = \frac{1}{\kappa} - d_1$ and $V(\ket{\psi}) = 1- \frac{1}{\kappa} - d_2$, then $\left|\braket{\chi}{\psi}\right|^2 \geq 1- \kappa d_1 - (\kappa + 1) \sqrt{\kappa \cdot d_2}$.
\end{lemma}
\begin{proof}
By \cref{lemma:non-negative_valid}, we know that $\left|\braket{\psi}{\phi}\right|^2 \geq 1 - \kappa\cdot d_2$. So by \cref{fact:preserv}, for any quantum state $\ket{\mu}$, $\left|\left|\braket{\mu}{\phi}\right|^2 - \left|\braket{\mu}{\psi}\right|^2\right| \leq \sqrt{\kappa \cdot d_2}$. We use this in two places. First, when $\ket{\mu} = \ket{+}$. Since $D(\ket{\psi}) = \left|\braket{+}{\psi} \right|^2 = \frac{1}{\kappa} - d_1$,  we have $\left|\braket{+}{\phi}\right|^2 \geq \frac{1}{\kappa} - d_1 - \sqrt{\kappa \cdot d_2}$ by triangle inequality. Second, when $\ket{\mu} = \ket{\chi}$. Notice that
\begin{align*}
\left|\braket{\chi}{\phi}\right|^2 = \kappa \left|\braket{+}{\phi}\right|^2 \geq 1 - \kappa(d_1 + \sqrt{\kappa\cdot d_2})\,.\end{align*}
Again by triangle inequality,
\begin{align*}
    \left|\braket{\chi}{\psi}\right|^2 \geq 1- \kappa (d_1 + \sqrt{\kappa \cdot d_2}) - \sqrt{\kappa \cdot d_2}\,.\tag*{\qedhere} 
\end{align*}
\end{proof}
Intuitively, \Cref{lemma:validform} allows us to \emph{tune} the probability of each test in the $\NP$ protocol. As we explain in the analysis (\Cref{sec:analysis-np}), if the probabilities of running \emph{Validity test} and \emph{Density test} are much higher than that for constraint tests, then if $d_1$ or $d_2$ is large, these two tests catch a ``deceptive'' quantum proof in soundness. This allows constraint tests to focus on the case of small $d_1$ and $d_2$; i.e. nearly \emph{rigid} quantum proofs.
\subsection{Constraint tests}
We analyze the constraint tests on \emph{rigid} quantum proofs, i.e. states of the form $\ket{\psi} = \frac{1}{\sqrt{R}} \sum_{j \in [R} \ket{j} \ket{\vec{v}_j}$. The verifier needs to check two properties:
\begin{enumerate}[label=(\roman*)]
    \item \label{cons1} (\emph{satisfiability}) For all $j\in [R]$, the assignment $\vec{v}_j$ satisfies $\calC_j$.
    \item \label{cons2} (\emph{consistency}) Each variable is assigned the same value when participating in different constraints.
\end{enumerate}
One may ask why we even need to check for \emph{consistency}. Couldn't we ask for the assignment of each variable $a: [N] \to \Sigma$, for example as the quantum proof $\frac{1}{\sqrt{N}} \sum_{i \in [N]} \ket{i}\ket{a(i)}$? The problem with this is checking \emph{satisfiability} becomes difficult, since the assigned values are given in superposition.\footnote{There is a way around this limitation for CSP systems consisting of \emph{unique game constraints}, where each (binary) constraint involving variables $i_1, i_2$ accepts exactly one $a(i_2)$ for each $a(i_1)$. See~\cite[Section 6]{qma2plus} for more discussion.}

Instead, with a state $\frac{1}{\sqrt{R}} \sum_{j \in [R} \ket{j} \ket{\vec{v}_j}$, \emph{satisfiability} is easy to verify: measure the first register (observing some $\ket{j}\ket{\vec{v}_j}$), and compute $\mathcal{C}_j(\vec{v}_j)$.
Let $u_s$ be the number of unsatisfied constraints. The outcome is $1$ with probability $1 - \frac{u_s}{R}$.

But this form of quantum proof gives Merlin a new way to ``deceive'': for a given variable, send different values depending on the constraint! We prevent this by checking for \emph{consistency}, similarly to the pre-processing step of \cite{dinur2007pcp} sometimes called \emph{regularization}. As in \cite[Section 7]{qma2plus}, we add ``consistency constraints'' to the CSP system $\mathcal{C}$ as follows:\footnote{Note that since $N$ and $R$ are polynomially-sized, this process is efficient.}
\begin{itemize}
    \item For each variable $i \in [N]$, let $V_i$ represent the constraints that $i$ participates in.
    \item Fix a constant $d$. For each $i \in [N]$, draw a $d$-regular graph with vertices $V_i$ that is \emph{expanding}.\footnote{For technical reasons of \Cref{claim:soundness-np}, we require that the Cheeger constant is at least 2.}
    \item Each edge $(j_1, j_2)$ of each expander $V_i$ represents a ``consistency constraint'', where we assert that the value of variable $i$ sent with constraint $j_1$ equals that sent with constraint $j_2$.
\end{itemize}
Using \emph{expander} graphs allows us to prevent this kind of ``deceptive'' Merlin: either the proof fails many of the original constraints, or it fails many ``consistency constraints''. Let $u_e$ be the number of unsatisfied ``consistency constraints'' out of $Rq \cdot \frac{d}{2}$:
\begin{claim}[{\cite[Lemma 4.1]{dinur2007pcp}}]\label{claim:soundness-np}
Consider a $(N,R,q,\Sigma)$-CSP system $\mathcal{C}$, and apply \emph{regularization}. If $\val(\calC) = 1$, then all ``consistency constraints'' can be simultaneously satisfied. If $\val(\calC) \le \delta$, then the total number of unsatisfied constraints $(u_s+u_e)$ is at least $(1-\delta) R$.
\end{claim}

How do we check these ``consistency constraints''?
Over the next few paragraphs, we construct a unitary related to permutations on the constraint graph. In completeness, the quantum proof is an eigenvector of this unitary, but in soundness, all \emph{rigid} quantum proofs are detectably far (i.e. using a Hadamard test) from an eigenvector.
We study the graph $\tilde{G}$ with $R \cdot q$ vertices, where each vertex $(j,i)$ corresponds to a clause $j$ and a variable $i$ that participates in clause $j$. Let $\tilde{G}$ be the union of all consistency edges created during regularization, i.e. $(j_1, j_2)$ for variable $i$ becomes the edge $((j_1,i), (j_2,i))$. Note that $\tilde{G}$ contains a copy of each expander graph, so it is $d$-regular.

We now choose $d$ permutations.
It is a classical fact that the adjacency matrix of any $d$-regular graph can always be decomposed to $d$ permutations. Let $\pi_1, \ldots, \pi_d$ be the decomposition of $\tilde{G}$; recall that these are permutations on $V(\tilde{G})$ where $|V(\tilde{G})| = R \cdot q$.
For each $k \in [d]$, we identify $\pi_k$ with a permutation on $[R] \times [N]$, where any $(j,i) \in [R] \times [N]$ that is not a vertex of $\tilde{G}$ (i.e. variable $i$ does not participate in constraint $j$) is mapped to itself.\footnote{These permutations (and their inverses) are all efficient because $N$ and $R$ are polynomially-sized.}
Note that this map always preserves the variable $i \in [N]$; without loss of generality, we also identify $\pi_k$ with its restriction $[R] \times [N] \to [R]$. From here on, we use this last definition of $\pi_k$, which maps constraint $j_1$ that variable $i$ participates in to another constraint $j_2$ that variable $i$ participates in, and identity otherwise.

Now consider a \emph{rigid} quantum proof, i.e. of the form $\ket{\psi} = \frac{1}{\sqrt{R}} \sum_{j \in [R]} \ket{j} \ket{\vec{v}_j}$. Since there are a polynomial number of variables and constraints, we can efficiently transform $\ket{\psi}$ to $\ket{\phi'}$, where
\begin{align*}
    \ket{\phi'} := \frac{1}{\sqrt{q\cdot  R}}\sum_{j\in [R]} \sum_{i \in \calC_j} \ket{j}\ket{\vec{v}_j}\ket{i}\ket{v_j(i)}\,.
\end{align*}
Here, $i \in \calC_j$ are the variables participating in $\calC_j$, and $v_j(i)$ is the value of this variable according to $\vec{v}_j$. 

We now would like to construct a unitary on $\ket{\phi'}$ that maps $\ket{j}\ket{i}\ket{value}$ to $\ket{j'}\ket{i}\ket{value}$ for some other constraint $j'$ that $i$ participates in. In completeness, this unitary would leave the state unchanged.
Notice that from the perspective of such a unitary, the second register containing $\ket{\vec{v}_j}$ is ``junk''.
Fortunately, we can measure out the second register in the Hadamard basis, and reject if the outcome is \emph{not} $\ket{+}$.
All \emph{rigid} states will observe outcome $\ket{+}$ with probability $\frac{1}{\kappa}$; one can see this by writing the second register in the Hadamard basis.

Suppose the observed outcome is $\ket{+}$; let us call the postselected state $\ket{\phi}$, where
\begin{align*}
\ket{\phi} := \frac{1}{\sqrt{q\cdot R}}\sum_{j \in R} \sum_{i \in \calC_j} \ket{j}\ket{i} \ket{v_j(i)}\,.
\end{align*}
For each $k \in [d]$, we now implement the in-place transformation $\Pi_k$ according to $\pi_k: [R] \times [N] \to [R]$, where
\begin{align*}
\Pi_k: \ket{j}\ket{i}\ket{v_j(i)} \rightarrow \ket{\pi_k(j,i)}\ket{i}\ket{v_j(i)}\,.
\end{align*}
Recall that the map $(j,i) \mapsto (\pi_k(j,i),i)$ is a permutation. Since we have access both to this permutation and its inverse, we can implement $\Pi_k$.

Note that in a satisfiable instance, $\Pi_k \ket{\phi} = \ket{\phi}$. By contrast, if $v_j(i) \neq v_{j'}(i)$, $\ket{j'}\ket{i}\ket{v_j(i)}$ is orthogonal to $\ket{\phi}$. Hence, $\bra{\phi}\Pi_k \ket{\phi}$ is the fraction of satisfied ``consistency constraints'' observed by $\pi_k$. We use the \emph{Hadamard test} to measure this value, in a similar way to the \emph{Spectral test} in~\cite{bassirian2022power}.
Note that unlike the swap test, the Hadamard test only uses one copy of a quantum state.
\begin{definition}[Hadamard test]
    Let $\ket{\psi}$ be a quantum state and $U$ a unitary operator.
    \begin{enumerate}
        \item Prepend a control qubit to $\ket{\psi}$, to create $\ket{0}\ket{\psi}$.
        \item Apply a Hadamard on the control qubit, to create $\frac{1}{\sqrt{2}}(\ket{0} + \ket{1}) \ket{\psi}$.
        \item Apply $U$, controlled by the control qubit, to create $\frac{1}{\sqrt{2}}\left( \ket{0}\ket{\psi} + \ket{1}U\ket{\psi} \right)$.
        \item Apply a Hadamard on the control qubit, to create $\frac{1}{2}\ket{0}\left(\ket{\psi} + U\ket{\psi} \right) + \frac{1}{2} \ket{1} \left(\ket{\psi} - U\ket{\psi} \right)$.
        \item Measure the control qubit, and accept if the output is $0$. 
    \end{enumerate}
    The success probability is then
    \begin{align*}
        \frac{1}{4} \left \| \ket{\psi} + U \ket{\psi} \right \|^2 = \frac{1}{2} + \frac{1}{4} \bra{\psi} U + U^\dagger \ket{\psi} = \frac{1}{2} + \frac{\Re\bra{\psi} U \ket{\psi}}{2} \,.
    \end{align*}
\end{definition}
We now can describe the constraint tests together:
\begin{enumerate}[label=(\roman*)]
\item \label{e:satisfied} With probability $\frac{1}{qd\kappa+1}$, check \emph{satisfiability}. This succeeds with probability $1-\frac{u_s}{R}$.
\item \label{e:consistent} With probability $\frac{qd\kappa}{qd\kappa+1}$, generate $\ket{\phi'}$, and measure the second register in the Hadamard basis. If the output state is not $\ket{+}$, reject. Otherwise, choose a random $k \in [d]$, and perform a Hadamard test with $\Pi_k$. This succeeds with probability $\frac{1}{\kappa}(\frac{1}{2} + \frac{1}{2}\mathbb{E}_k[\Re \bra{\phi}\Pi_k \ket{\phi}]) = \frac{1}{\kappa}(1- \frac{u_e}{qdR})$.\footnote{Note that in our protocol, $\bra{\phi}\Pi_k \ket{\phi}$ is always real because $\ket{\phi}$ and $\Pi_k$ have real values.}
\end{enumerate}
The overall success probability of the constraint tests is 
\begin{align*}
\frac{1}{qd\kappa+1} (1 - \frac{u_s}{R}) + \frac{qd\kappa}{qd\kappa + 1}\left(\frac{1}{\kappa} \cdot (1 - \frac{u_e}{qdR}) \right)
=
\frac{qd + 1}{qd\kappa + 1}- \frac{u_e + u_s}{R \cdot (qd\kappa+1)}\,.
\end{align*}
We now show a constant gap between completeness and soundness.
In completeness, $u_e = u_s = 0$, so $\ket{\psi}$ passes the constraint tests with probability $C^{\textnormal{YES}} := \frac{qd + 1}{qd \kappa + 1}$. In soundness, recall that $\val(\mathcal{C}) \le \delta$, so by \Cref{claim:soundness-np}, any \emph{rigid} quantum proof passes the constraint tests with probability at most $C^{\textnormal{YES}} - \frac{1-\delta}{qd\kappa + 1}$. We now apply \Cref{lemma:validform}: any quantum proof that passes \emph{Density test} and \emph{Validity test} with probabilities too similar to that in completeness must pass the constraint tests with probability less than $C^{\textnormal{YES}}$.
\begin{corollary}
\label{cor:faillabeling}
In soundness, if $D(\ket{\psi}) = \frac{1}{\kappa} - d_1$ and $V(\ket{\psi}) = 1-\frac{1}{\kappa} - d_2$, then
\begin{align*}C(\ket{\psi}) \leq C^{\textnormal{YES}} - \frac{1-\delta}{qd \kappa + 1} + \left(\kappa d_1 + (\kappa+1)\sqrt{\kappa \cdot d_2}\right)^{1/2}\,.
\end{align*}
\end{corollary}

\subsection{Analysis} \label{sec:analysis-np}
In the protocol, Arthur applies \emph{Density test}, \emph{Validity test}, or \emph{constraint tests} with probability $p_1, p_2, p_3$, respectively, where $p_3 = 1-p_1-p_2$.

We start by analyzing the success probability of the protocol in completeness. Here, $\val(\calC) =1$, and the quantum proof $\ket{\psi} = \frac{1}{R}\sum_j \ket{j} \ket{\vec{v}_j}$ is such that $\vec{v}_j$ is a satisfying assignment to the variables that participate in $\calC_j$. The success probability for each test is exactly $\frac{1}{\kappa}$, $1 - \frac{1}{\kappa}$, and $C^{\textnormal{YES}}$, respectively.
So the success probability of the protocol in completeness is $P_{\textnormal{YES}} = \frac{p_1}{\kappa}+p_2(1 - \frac{1}{\kappa}) + p_3 C^{\textnormal{YES}}$. 

We now choose the probabilities $p_1,p_2,p_3$. Choose $\lambda := \frac{1-\delta}{qd\kappa + 1}$.
\begin{enumerate}
\item We first set a \emph{distance threshold} $\eps := \frac{\lambda}{\Gamma}$ for a large enough constant $\Gamma(\kappa, q, d, \delta)$ satisfying
\begin{align*}
 \left(\kappa\eps + (\kappa+1)\sqrt{\kappa \cdot \eps} \right)^{1/2} \leq \frac{\lambda}{2}\,.
\end{align*}
\item Let $Z :=\frac{1}{2} + 1 +  \frac{\eps}{4(1-C^{\textnormal{YES}})}$. Then let 
\begin{align*}
p_1 &= \frac{1}{2} \cdot \frac{1}{Z} & p_2 &= \frac{1}{Z} & p_3 &= \frac{\eps}{4(1-C^{\textnormal{YES}})} \cdot \frac{1}{Z}\,.
\end{align*}
\end{enumerate}
Now we study soundness, i.e. when $\val(\calC) \leq \delta$. We again denote the quantum proof as $\ket{\psi}$. We divide up the analysis into a few parts:
\begin{enumerate}
    \item A quantum proof that is ``too sparse'' (i.e.\ $D(\ket{\psi}) = \frac{1}{\kappa} - d$ for any $d \geq \eps$) is detected by \emph{Density test}.
    \begin{align*}
        P_{\textnormal{NO}} &= p_1(\frac{1}{\kappa} - d) + p_2V(\ket{\psi}) + p_3C(\ket{\psi}) 
        \\
        &\le P_{\textnormal{YES}} - p_1 d + p_3 (1-C^{\textnormal{YES}}) 
        \\
        &\le P_{\textnormal{YES}} - p_1 \eps + p_3 (1-C^{\textnormal{YES}})
        \\
        &= P_{\textnormal{YES}} - \frac{\eps}{2Z} + \frac{\eps}{4Z} = P_{\textnormal{YES}} - \frac{\eps}{4Z}\,.
    \end{align*} 
    \item A quantum proof that is ``too dense'' (i.e.\ $D(\ket{\psi}) = \frac{1}{\kappa} + d$ for any $d \geq \eps$) is detected by \emph{Validity test}.
    \begin{align*}
        P_{\textnormal{NO}} &= p_1 (\frac{1}{\kappa} + d) + p_2  V(\ket{\psi}) + p_3 C(\ket{\psi})
        \\
        &\le p_1 (\frac{1}{\kappa} + d) + p_2  (1 - \frac{1}{\kappa} - d) + p_3
        \\
        &= P_{\textnormal{YES}} - (p_2 - p_1) d + p_3 (1-C^{\textnormal{YES}})  
        \\
        &\le P_{\textnormal{YES}} - (p_2 - p_1) \eps + p_3  (1-C^{\textnormal{YES}}) \\
        &= P_{\textnormal{YES}} - \frac{\eps}{2Z} + \frac{\eps}{4Z} = P_{\textnormal{YES}} - \frac{\eps}{4Z}\,,
    \end{align*}
    where the first inequality follows from \cref{lemma:density_validity}.
    \item A quantum proof that is ``the right density''  (i.e. $D(\ket{\psi}) = \frac{1}{\kappa} + d_1$ for $|d_1| \leq \eps$) but far from ``valid'' ($V(\ket{\psi}) = 1 -  \frac{1}{\kappa} - d_2$ for $d_2 \geq \eps$) is detected by \emph{Validity test} when $p_2 > p_1$.
    \begin{align*}
        P_{\textnormal{NO}} &\le p_1 (\frac{1}{\kappa} + |d_1|) + p_2 (1 - \frac{1}{\kappa} - d_2) + p_3
        \\
        &\le P_{\textnormal{YES}} + p_1  |d_1| - p_2 d_2 + p_3 (1-C^{\textnormal{YES}})  \\
        &\le P_{\textnormal{YES}} - (p_2 - p_1) \eps +  p_3 (1-C^{\textnormal{YES}}) \\
        &= P_{\textnormal{YES}} - \frac{\eps}{4Z}\,.
    \end{align*}
    \item Lastly, a quantum proof that is nearly \emph{rigid} (i.e. $D(\ket{\psi}) = \frac{1}{\kappa} + d_1$ and $V(\ket{\psi}) = 1 - \frac{1}{\kappa} - d_2$ for any $|d_1|, d_2 \leq \eps$) is detected by the constraint tests.
    \begin{align*}
        P_{\textnormal{NO}} &= p_1 (\frac{1}{\kappa} + d_1) + p_2 (1 - \frac{1}{\kappa} - d_2) + p_3   C(\ket{\psi})
        \\
        &\le P_{\textnormal{YES}} + p_1 d_1 - p_2 d_2 + p_3 \left(- \frac{1-\delta}{qd\kappa + 1} + \left(\kappa d_1 + (\kappa+1)\sqrt{\kappa \cdot d_2}\right)^{1/2} \right) \\
        &\le P_{\textnormal{YES}} + p_1 d_1 - p_2 d_2 - p_3\frac{\lambda}{2} \,.
    \end{align*}
    The first inequality follows from \cref{cor:faillabeling}, and the second inequality holds by our choice of $\eps$. Note that $d_2 \ge 0$ by \Cref{lemma:validity_upper_bound}. By \cref{lemma:density_validity}, if $d_1 \ge 0$, then $d_1 \le d_2$; otherwise $d_1 \le d_2$ trivially. So then
    \begin{align*}
        P_{\textnormal{NO}} 
        &\le P_{\textnormal{YES}} - (p_2 - p_1) d_2 - p_3 \frac{\lambda}{2} 
        \le P_{\textnormal{YES}} - \frac{\eps \lambda}{8 (1-C^{\textnormal{YES}}) Z}\,.
    \end{align*}
\end{enumerate}
Combining these cases proves the following result:
\begin{theorem}
\label{thm:ugc_protocol_qmas}
    Given an instance of $(1, \delta)$-$\GapCSP$, the $\QMA^+_{\log}$ protocol succeeds with probability $P^{\textnormal{YES}}$ in completeness and at most $P^{\textnormal{YES}} - \Delta$ in soundness for some constants $1 > P^{\textnormal{YES}} > \Delta > 0$.
\end{theorem}

\begin{corollary}
    There exist constants $1 > P^{\textnormal{YES}} > \Delta > 0$ such that $\NP \subseteq \QMA^+_{\log}$ with completeness $P^{\textnormal{YES}}$ and soundness $P^{\textnormal{YES}} - \Delta$.
\end{corollary}

\section{$\QMA^+$ protocol for $\NEXP$}
\label{section:nexp}
Our goal in this section is to modify the previous protocol to solve an $\NEXP$-complete problem. Again by the PCP theorem, the \emph{succinct} $(1, \delta)$-$\GapCSP$ problem with exponentially many variables and clauses is $\NEXP$-complete. The \emph{succinctness} allows us to efficiently describe the problem input. 
What remains is to ensure that the verifier's protocol is efficient. Previously, the unitary transformations were efficient because the verifier handled $\poly(n)$-size graphs. Furthermore, the expanders used to check the equality constraints for each variable may have different sizes. Now that there can be exponentially many possibilities for the size of each cluster, naively applying the previous technique is not efficient. These challenges were addressed in~\cite{qma2plus} by considering a PCP construction for $\NEXP$ with strong properties.
\begin{theorem}[\cite{qma2plus}]
There is a $\NEXP$-hard $(1, \delta)$-$\GapCSP$ instance for some $(N = 2^{\poly(n)}, R = 2^{\poly(n)}, q = O(1), \Sigma = \{0,1\})$-CSP system $\mathcal{C}$ that is both \emph{$\tau$-strongly uniform} for some constant $\tau$ and $\polylog(NR)$-\emph{doubly explicit}.
\end{theorem}
Informally, every constraint in a \emph{succinct} CSP system must be computable in polynomial time. The \emph{doubly explicit} property further requires the existence of efficient maps from variables to constraints \emph{and} from constraints to variables. Intuitively, these maps allow us to efficiently implement the Hadamard test of the consistency checks.

We include the formal definition of these properties.
Define $\Adj_\calC(j)$ to be the list of variables participating in $\calC_j$, and $\Adj_V(i)$ be the list of constraints that depend on variable $i$.
 \begin{definition}[Doubly explicit CSP]
     A $(N, R, q, \Sigma)$-CSP system $\calC$ is $Z(N,R)$-\emph{doubly explicit} if for all $i \in [N]$ and $j \in [R]$, the following are computable in time $Z(N,R)$:
     \begin{enumerate} [label=(\roman*)]
     \item Cardinality of $\Adj_V(i)$ and $\Adj_\calC(j)$ for all $i \in [N]$ and $j \in [R]$.
     \item $\Adj_\calC^{\text{ind}}: [R]\times [N] \rightarrow [q]$; if $i$ participates in $\calC_j$, then $\Adj_\calC^{\text{ind}}(j, i) = \imath$ is the index of $i$ in $\Adj_\calC(j)$.
     \item $\Adj_\calC^{\text{id}}: [R] \times [q] \rightarrow [N]$; $\Adj_\calC^{\text{id}}(j, \imath) = i$ is the $\imath$-th variable of $\Adj_\calC(j)$.
     \item $\Adj_V^{\text{ind}}: [N] \times [R] \rightarrow [R]$; if $i$ participates in $\calC_j$, then $\Adj_V^{\text{ind}}(i, j)=\jmath$ is the index of $j$ in $\Adj_V(i)$.
     \item $\Adj_V^{\text{id}}: [N] \times [R] \rightarrow [R]$; $\Adj_V^{\text{id}}(i, \jmath) = j$ is the $\jmath$-th variable $\Adj_V(i)$.
     \end{enumerate}
 \end{definition}
 This property alone is not enough for efficient regularization: the verifier must know how to implement an expander of size $\left|\Adj_V(i)\right|$ for \emph{all} variables $i$. The \emph{strongly uniform} property resolves this complication. 
 \begin{definition}[Strongly uniform CSP]
 Let $\tau \in \mathbb{N}$. A $(N, R, q, \Sigma)$-CSP system $\calC$ is $\tau$-strongly uniform if the variable set $[N]$ can be partitioned into at most $\tau$ different subsets $\bigcup_y V_y$ such that $|\Adj_V(i)| = |\Adj_V(j)| = n_k$ if $i$ and $j$ belong to the same part $V_k$. Furthermore, the part $k \in [\tau]$ can be determined in time $\polylog(NR)$.
 \end{definition}
 A $\tau$-strongly uniform CSP system allows the verifier to use $\tau$ different (possibly exponential size) $d$-regular expanders. These can be constructed in polynomial time: 
\begin{theorem}[Doubly explicit expander graphs \cite{lubotzky2009finite, alon2021explicit}]
There is a constant $d$ such that the following explicit constructions of expander graphs exist:
\begin{enumerate}
    \item \label{exp:1} For every $n$, there is a $d$-regular graph on $n$ vertices.
    \item \label{exp:2} For every prime $p > 17$, there is a $d$-regular graph on $n = p(p^2 - 1)$ vertices, and the graph can be decomposed into $d$ permutations $\pi_1, \dots, \pi_d$ that can each be evaluated in time $\polylog(n)$.\footnote{In fact, since these graphs are Cayley graphs, both the permutations and their inverses can be evaluated in time $\polylog(n)$. We use both $\pi_k$ and $\pi_k^{-1}$ in the constraint tests to implement the unitary $\calM_k$.}
\end{enumerate}
Furthermore, the neighbors of each variable can be listed in $\polylog(n)$, and the graphs have Cheeger constant at least 2.
\end{theorem}
With this theorem, the verifier can choose a large constant $n_0$, and use Construction \ref{exp:1} if $n_i \le n_0$. Otherwise, the verifier can cover \emph{almost} all $n_i$ vertices with an explicit expander using Construction \ref{exp:2}.
\begin{theorem}[Primes in short intervals \cite{cheng_primes}]
There is an absolute constant $k_0$ such that for any integer $k > k_0$, there is a prime in the interval $[k - 4k^{2/3}, k]$.
\end{theorem}
We modify the protocol to expect the quantum proof $\ket{p_1,p_2, \dots, p_\tau} \otimes \frac{1}{ R}\sum_{j \in R} \ket{j}\ket{\vec{v}_j}$ such that each 
$p_i \in [\lfloor n_{i}^{1/3}\rfloor - 4 \lfloor n_{i}^{1/3}\rfloor^{2/3}, \lfloor n_{i}^{1/3}\rfloor]$. The verifier measures the primes, and can always check that every $p_i$ is a prime number in the required range.
The rest of the analysis is similar to the $\NP$ protocol, but regularized using these efficient expanders.
We first explain how to efficiently implement the \emph{constraint tests}, and then analyze the $\QMA^+$ protocol.

\subsection{Efficient constraint tests}
We show how to efficiently implement  \emph{consistency} checks that imply a version of \cref{claim:soundness-np}.
Fix any vertex $i$. Let $n$ be the number of constraints that depend on $i$ and $p$ be the corresponding prime. Let $n_0$ be a large enough constant. If $n \le n_0$, then use Construction \ref{exp:1} $d$-regular expander to wire new copies of the vertex together, just as for $\NP$. Otherwise, use Construction \ref{exp:2} to generate a $d$-regular expander graph of size $p(p^2-1) \in [n - O(n^{8/9}), n]$ that wires nearly all copies of the vertex together. Then add $d$ self-loops for the remaining vertices. The number of vertices with self loops is at most some $\eta n$ (for some small constant $\eta$) since $p(p^2-1) \geq n - O(n^{8/9})$; we can make $\eta$ arbitrarily small by choosing a large enough $n_0$.

\begin{claim}[\cite{qma2plus}]\label{claim:soundness}
Consider a $(N,R,q,\Sigma)$-CSP system $\mathcal{C}$, and apply \emph{efficient regularization}. If $\val(\calC) = 1$, then all ``consistency constraints'' can be simultaneously satisfied. If $\val(\calC) \le \delta$, then the total number of unsatisfied constraints is at least $(1-\delta-q\eta) R$.
\end{claim}
The analysis of \Cref{claim:soundness} is similar to \cref{claim:soundness-np}. The additional factor of $q\eta$ comes from the self-loop constraints; these can be satisfied without violating any ``consistency constraint''. 

After measuring the primes, let the verifier act on the space $\mathbb{C}^R \otimes \mathbb{C}^{\kappa} \otimes \mathbb{C}^N \otimes \mathbb{C}^{|\Sigma|}$. We explicitly define the unitary operators that are used in the $\NP$ protocol. These definitions exactly match~\cite{qma2plus}. First, operator $\calA$ expands the values from the list of values of variables involved in each constraint:
\begin{align*}
    \calA: \ket{j}\ket{v}\ket{0}\ket{0} &\rightarrow \frac{1}{q} \sum_{r=1}^q \ket{j}\ket{\vec{v}}\ket{i_r}\ket{v_r}\,,
\end{align*}
Here, $v$ is the list of values of variables involved in constraint $j$,  $i_r$ is the $r$-th variable involved in $j$, and $v_r$ is the value of $i_r$ according to $v$. 
Next, define the permutation operators $\calM_k$ for each $k \in [d]$ that implement the $d$ permutations of each efficiently constructed expander:
\begin{align*}
    \calM_k: \ket{j}\ket{v} \ket{i}\ket{v'} \rightarrow \ket{\Pi_k(j, i)} \ket{v} \ket{i} \ket{v'}
\end{align*}
The last operation computes the constraints in superposition:
\begin{align*}
    \calB \ket{j}\ket{v} \ket{0} \rightarrow \ket{j}\ket{v}\ket{\calC_j(v)}
\end{align*}

\begin{theorem}[\cite{qma2plus}]
$\calA, \calB, \calM_k$ can be implemented by $\BQP$ circuits.
\end{theorem}
\subsection{Analysis}
The analysis is nearly the same as in the $\NP$ protocol, with the minor difference that the verifier also receives $p_1, \ldots, p_\tau$ in the quantum proof. The rigidity tests are unchanged. For the constraint tests, the verifier can use the explicit operators $\calA, \calB, \calM_k$:
\begin{enumerate}[label=(\roman*)]
\item For \emph{satisfiability}, the prover computes $\calB \ket{\psi}\ket{0}$ and measures the second qubit in the standard basis.
\item For \emph{consistency}, the prover computes $\calA \ket{\psi}$, selects random $d \in [k]$, and uses $\calM_k$ in Hadamard test.
\end{enumerate}
We already know that for a quantum proof of the valid form, we can write the success probability as:
\begin{align*}
    C(\ket{\psi}) &= \frac{qd+1}{qd\kappa+1} - \frac{u_e+u_s}{R\cdot(qd\kappa+1)}
\end{align*}

To be able to analyze soundness, all that is left is to reprove \cref{cor:faillabeling} to handle the subtle difference between \cref{claim:soundness} and \cref{claim:soundness-np}.
Let $D(\ket{\psi}) = \frac{1}{\kappa} \pm d_1$ and $V(\ket{\psi}) = 1 - \frac{1}{\kappa} - d_2$, then we can write:
\begin{align*}
C(\ket{\psi}) &=  C^{\textnormal{YES}} - \frac{u_e+u_s}{R\cdot(qd\kappa+1)}\\
& \leq C^{\textnormal{YES}} - \frac{R\cdot(1-\delta-q\eta)}{R \cdot(qd\kappa+1)} + \left(\kappa d_1 + (\kappa +1)\sqrt{\kappa \cdot d_2}\right)^{1/2}\,.
\end{align*}
We can choose a small enough $\eta$ (via large enough $n_0$) so that $\eta q < \frac{1-\delta}{2}$. 
\begin{corollary} \label{corollary:soundness-constraints}
If $D(\ket{\psi}) = \frac{1}{\kappa} \pm d_1$ and $V(\psi) = 1-\frac{1}{\kappa} - d_2$, then:
$$C(\ket{\psi}) \leq C^{\textnormal{YES}} - \frac{1-\delta}{2\left(qd\kappa+1\right)} + \left(\kappa d_1 + (\kappa+1)\sqrt{\kappa \cdot d_2}\right)^{1/2}$$
\end{corollary}
For soundness, the same analysis of \cref{sec:analysis-np} goes through by reducing $\lambda$ by a factor of $2$; this change comes from \cref{corollary:soundness-constraints}, which handles the extra $q\eta$ self-loop constraints. All together:
\begin{theorem}    Consider $(1,\delta)$-$\GapCSP$ with a $(N = 2^{\poly(n)}, R = 2^{\poly(n)}, q = O(1), \Sigma = \{0,1\})$-CSP system that is $\polylog(NR)$-doubly explicit and $O(1)$-strongly uniform. The $\QMA^+$ protocol solves this problem with completeness $c$ and soundness $s$ for some constants $1 > c > s > 0$.
\end{theorem}
\begin{corollary}
    There exist constants $1 >c > s >  0$ such that $\NEXP \subseteq \QMA^+_{c,s}$.
\end{corollary}

\section{Subtle features of $\QMA^+$}
\label{sec:qmaplus_subtleties}
\subsection{Promise symmetry matters}
One can imagine restricting to proofs with non-negative amplitudes \emph{only} in completeness. But this class is equal to $\QMA$:

\begin{fact}
\label{fact:qma_mod_YES_only}
    Consider the class $\QMA^{+'}$, where the proof must have non-negative amplitudes \emph{only in completeness}. 
Since subset states have non-negative amplitudes, $\SQMA \subseteq \QMA^{+'} \subseteq \QMA$. By \cite{sqma}, $\SQMA = \QMA^{+'} = \QMA$.
\end{fact}

Instead, $\QMA^+$ also restricts the proof in soundness, which reduces the ways Merlin can ``deceive'' Arthur.
This increases the power of the complexity class:
\begin{corollary}
\label{cor:qma_at_most_qmaplus}
    Notice that $\QMA^{+'}_{c,s} \subseteq \QMA^+_{c,s}$ for any choice of $0 \le c,s \le 1$, since any $\QMA^{+'}_{c,s}$ protocol is also a $\QMA^+_{c,s}$ protocol.
    Since $\QMA = \QMA^{+'}$, $\QMA \subseteq \QMA^+_{c,s}$ whenever $c < 1$ and $c - s \ge  \frac{1}{p(n)}$ for any polynomial $p(n)$.
\end{corollary}
In general, suppose $\mathcal{R}$ is a set of quantum states that approximate all quantum states (i.e. an $\epsilon$-covering) by at least an inverse polynomial in number of qubits. Then $\QMA$ is equal to $\QMA$ restricted to $\mathcal{R}$ in completeness, and at most $\QMA$ restricted to $\mathcal{R}$ in both completeness and soundness.

Furthermore, classes that modify $\QMA$ only in completeness enjoy promise gap amplification through parallel repetition. This does not hold for promise-symmetric modifications.
We provide a simple example of how parallel repetition fails to amplify the promise gap of $\QMA^+$:
\begin{fact}
\label{fact:max_of_reals_is_not_norm}
    Consider a Hermitian and positive semidefinite matrix $M$.
Let $\| M \|_+ := \max_{\ket{v} \ge 0} \frac{\| M \ket{v} \|_2}{\|\ket{v}\|_2}$ be the maximum value among real non-negative vectors. Then it is possible for $\| M \otimes M \|_+ > \| M \|_+^2$.
\end{fact}
\begin{proof}
Consider two qubits and the projector $M = \ket{x_-}\bra{x_-}$, where $M$ projects into the Pauli-X basis (i.e. $\ket{x_-} = \frac{1}{\sqrt{2}}(\ket{0} - \ket{1})$). 
Then  $\| M \|_+ = \frac{1}{\sqrt{2}}$, maximized at $\ket{0}$ or $\ket{1}$.
But using the state $\ket{\chi} = \frac{1}{\sqrt{2}}(\ket{00} + \ket{11})$, $\| M \otimes M \|_+ \ge \| M \ket{\chi}\| = \frac{1}{\sqrt{2}} > \| M \|_+^2$.
\end{proof}

\subsection{$\QMA^+$ at some constant gap equals $\QMA$}
Perhaps surprisingly, $\QMA^+_{c,s}$ \emph{equals} $\QMA$ for some constants $1 > c > s > 0$. This is because every quantum state can be approximated (up to a constant) by a quantum state \emph{without relative phase}:
\begin{proposition}
\label{prop:qmaplus_nearly_qma}
$\QMA^+_{c,s} \subseteq \QMA_{c,4s}$.
\end{proposition}
\begin{proof}
Consider a problem in $\QMA^+_{c,s}$, and let $\Pi_1$ be its accepting operator. We will use the same circuit in $\QMA$, and analyze the new completeness and soundness.
    \begin{itemize}
        \item \emph{(completeness)} The same completeness proof is a valid proof for $\QMA$, accepting with completeness $c$.
        \item \emph{(soundness)} Recall that $\bra{\chi}\Pi_1 \ket{\chi} \le s$ for any $\ket{\chi}$ with non-negative amplitudes. 
        
        Consider any state $\ket{\psi}$ with \emph{real} (but possibly negative) amplitudes. Separate and normalize its positive entries and negative entries, i.e. $\ket{\psi} = \sqrt{p} \ket{\psi_+} - \sqrt{1-p} \ket{\psi_-}$. Notice that $\bra{\psi_+}\ket{\psi_-} = 0$, so $\ket{\psi}$ is of unit norm.
        Then
        \begin{align*}
        \bra{\psi}\Pi_1\ket{\psi} 
        &= p \bra{\psi_+} \Pi_1 \ket{\psi_+} + (1-p)  \bra{\psi_-} \Pi_1 \ket{\psi_-} - \sqrt{p(1-p)} ( \bra{\psi_-} \Pi_1 \ket{\psi_+}  +  \bra{\psi_+} \Pi_1 \ket{\psi_-} ) 
        \\ &\le p \bra{\psi_+} \Pi_1 \ket{\psi_+} + (1-p)  \bra{\psi_-} \Pi_1 \ket{\psi_-} + 2\sqrt{p(1-p)} \sqrt{ \bra{\psi_-} \Pi_1 \ket{\psi_-} \bra{\psi_+} \Pi_1 \ket{\psi_+}}
        \\ &\le s + 2s\sqrt{p(1-p)} \le 2s\,,
        \end{align*}
        where the first inequality holds by Cauchy-Schwarz because $\Pi_1$ is positive semidefinite.

        Similarly, consider any state with arbitrary amplitudes. Separate and normalize its real and imaginary entries, i.e. $\ket{\phi} = \sqrt{p'}\ket{\phi_{\mathbb{R}}} + i \sqrt{1-p'} \ket{\phi_{i\mathbb{R}}}$. Notice that $\ket{\phi}$ is still unit norm. By the same calculation, one finds that $\bra{\phi} \Pi_1 \ket{\phi} \le 4s$.
    \end{itemize}
\end{proof}
\begin{corollary}
\label{cor:qma_plus_equals_qma}
For any $0 < \eps < 0.2$, $\QMA^+_{0.8+\eps,0.2} = \QMA$.
\end{corollary}
\begin{proof}
    $\QMA^+_{0.8+\eps,0.2} \subseteq \QMA$ follows from \Cref{prop:qmaplus_nearly_qma}. The other direction follows from \Cref{cor:qma_at_most_qmaplus}.
\end{proof}
This is a strange phenomenon: depending on the choice of constants $c>s$, $\QMA^+_{c,s}$ could be as small as $\QMA$ and as large as $\NEXP$!\footnote{
Note that the same phenomenon holds for $\QMA^+(2)$ and $\QMA(2)$ with nearly the same proof. This is why \cite{qma2plus} was perceived as ``just a constant gap away'' from proving $\QMA(2) = \NEXP$.}
See \Cref{fig:qma_vs_nexp} for a pictorial description.
An implication of our work is that assuming $\EXP \ne \NEXP$, $\QMA^+$ simply \emph{cannot} be amplified.


\section{Open questions}
\begin{enumerate}
    \item \textbf{What is the relationship of $\QMA$ and $\QMA(2)$?} Our result does not immediately say anything about $\QMA$ and $\QMA(2)$. It only suggests that for $\QMA$, the restriction of \emph{relative phase} is maximally strong. For example, it is possible that $\QMA(2) = \NEXP$; i.e. the restriction of \emph{entanglement} across a fixed barrier may be just as powerful.
In fact, showing $\QMA(2) = \QMA^+(2)$ is still an open route to proving $\QMA(2) = \NEXP$, but in light of this work, amplification for $\QMA^+(2)$ must crucially rely on the unentanglement promise.
\item \textbf{Are other complexity classes sensitive to different constant-sized promise gaps?} We show that for $\QMA^+_{c,s}$, the parameter $\frac{c}{s}$ can be ``tuned'' to change the power of the class from $\QMA$ to $\NEXP$ (see also \Cref{fig:qma_vs_nexp}). Do other complexity classes drastically change power with different promise gaps? One similar class is $\SBQP$~\cite{kuperberg2015hard}, which equals $\BQP_{c,s}$ when $\frac{c}{s} = 2$ (but unlike our work, $c$ and $s$ are exponentially small). However, when $\frac{c}{s}$ is allowed to be any number above $1$, $\BQP_{c,s}$ is equal to $\PP$~\cite{deshpande2022importance}. Note that relative to oracles, $\SBQP$ is not closed under intersection, which was used to separate it from $\QMA$~\cite{sbp_qma}. 
\item \textbf{State complexity vs. decision complexity?} Although we prove that there are constants $c_1,s_1,c_2,s_2$ such that $\QMA^+_{c_1,s_1} = \QMA^+(2)_{c_2,s_2}$, we do not prove the existence of any \emph{product test} (as in~\cite{harrow2013testing}). In fact, it is possible that no product test exists! This would show a separation between the complexity of \emph{decision problems} and \emph{state synthesis problems}; i.e. $\QMA^+ = \QMA^+(2)$ but $\stateQMA^+ \ne \stateQMA^+(2)$. In fact, it is even possible that $\QMA = \QMA(2)$ but $\stateQMA \ne \stateQMA(2)$. This inquiry can help us understand whether (or how) \emph{the power of unentanglement} is useful when solving \emph{decision problems}.
\end{enumerate}

\section*{Acknowledgements}
Thanks to Zachary Remscrim for collaborating on early
stages of this project. 
Thanks to Noam Lifshitz, Dor Minzer, and Kevin Pratt for answering questions about algebraic constructions of expanders.
Thanks to Srinivasan Arunachalam, Fernando Granha Jeronimo, Supartha Podder, and Pei Wu for comments on a draft of this manuscript.

BF and RB acknowledge support from AFOSR (award number FA9550-21-1-0008). This material is based upon work partially supported by the National Science Foundation under Grant CCF-2044923 (CAREER) and by the U.S. Department of Energy, Office of Science, National Quantum Information Science Research Centers as well as by DOE QuantISED grant DE-SC0020360. KM acknowledges support from the National Science Foundation Graduate Research Fellowship Program under Grant No. DGE-1746045. Any opinions, findings, and conclusions or recommendations expressed in this material are those of the author(s) and do not necessarily reflect the views of the National Science Foundation.

\bibliography{_references}
\bibliographystyle{_alpha-betta-url}

\end{document}